\documentclass[]{article}
\usepackage[english]{babel}
\usepackage[colorlinks=true, allcolors=blue]{hyperref}
\usepackage[autostyle,italian=guillemets]{csquotes}
\usepackage[backend=biber,style=numeric]{biblatex}
\usepackage{amssymb}
\usepackage{mathtools}
\usepackage{amsthm}
\usepackage[section]{algorithm}
\usepackage{algpseudocode}

\theoremstyle{plain}
\newtheorem{theorem}{Theorem}[section]
\newtheorem{lemma}[theorem]{Lemma}

\newtheorem{corollary}[theorem]{Corollary}

\theoremstyle{definition}
\newtheorem{definition}[theorem]{Definition}

\numberwithin{equation}{section}

\newcommand{\N}{\mathbb{N}}

\title{How to program a never-losing chess engine}
\author{Fabio Romano}

\begin{document}

\maketitle

\begin{abstract}
	This article proposes a model, based on graph theory, to represent a variety of two-player games of perfect information, such as chess and checkers. I then provide a backtracking minimax algorithm to find, if it exists, a perfect game strategy (game resolution), and subsequently a way to exploit that algorithm to determine a weaker condition: the existence of a strategy to never lose (always reach at least a draw).
	
	Of course, this does not mean that in practice the algorithm can find such a strategy in a short time, but here we are only concerned with formally proving that this is possible, at least theoretically.
\end{abstract}

\section{The Shannon's idea}

In 1950, Claude Shannon informally outlined a theoretical procedure for playing a perfect game:
\blockquote[Claude Shannon]
{With chess it is possible, in principle, to play a perfect game or construct a machine to do so as follows: One considers in a given position all possible moves, then all moves for the opponent, etc., to the end of the game (in each variation). The end must occur, by the rules of the games after a finite number of moves (remembering the 50 move drawing rule). Each of these variations ends in win, loss or draw. By working backward from the end one can determine whether there is a forced win, the position is a draw or is lost.}

The above idea is the foundation for Shannon's work \cite{shannon:computer_playing_chess} of 1950, which is fundamental, but, theoretically speaking, this article has several issues. First of all, it focuses exclusively on the game of chess, while we would like to design a more general procedure. Second, it relies mostly on heuristics of evaluating functions, lacking of formal proofs for the proposed strategy, while we would like to have a strategy that is provably sound. This is what I intend to do in this article.

\section{Graph terminology and the game model}

Let $G = \left< V, E \right>$ be a directed graph.

\begin{definition}[Path]\label{def:path}
	A path is a sequence of vertices $\left< v_1, \dots, v_n \right> \in V^n$, with $n \ge 1$, such that $(v_i, v_{i+1}) \in E$ for every $i \in \{1, \dots, n-1\}$.
\end{definition}

\begin{definition}[Reachability]\label{def:reachability}
	A vertex $v \in V$ is said to be reachable from a vertex $v_1 \in V$ if there exists a path $\left< v_1, \dots, v_n \right>$ such that $v_n = v$ for some $n \in \N$. In that case, we write $v_1 \leadsto v$.
\end{definition}

\begin{definition}[Cycle]
	A cycle is a path $\left< v_1, \dots, v_n \right>$, with $n \ge 2$, such that $v_1 = v_n$.
\end{definition}

\begin{definition}[Simple path]
	A simple path is a path $\left< v_1, \dots, v_n \right>$ which does not contain cycles, i.e. $\forall i,j \in \{1, \dots, n\}, i \ne j \colon v_i \ne v_j$.
\end{definition}

\begin{definition}[Sinks]
	In this context, we call "sink" those vertices of $G$ which the only arc they have is a self-loop. Formally:
	\[
	sinks(G) := \{v \in V \mid G.Adj(v) = \{v\} \}
	\]
	
	where $G.Adj(v)$ is the set of vertices that are adjacent to $v$, i. e. $G.Adj(v) := \{u \in V \mid (v, u) \in E\}$.
\end{definition}

\begin{definition}[Game graph] \label{def:game_graph}
	A game graph $G := \left< V, E, v_0, win, turn \right>$ is a directed graph in which:
	
	\begin{itemize}
		\item $V$ is the set of positions that can occur during the game;
		
		\item $E \subseteq V^2$ is the set of legal moves;
		
		\item $v_0 \in V \setminus sinks(G)$ is the initial position of the game;
		
		\item $\forall v \in V \setminus sinks(G) \colon G.Adj(v) \ne \emptyset \text{ and } v \notin G.Adj(v)$.
		
		\item $win \colon sinks(G) \to \{-1, 0, 1\}$ is a function that returns the winner for every endgame position: 1 if the attacker won, -1 if the defender won, 0 if it is a draw;
		
		\item $turn \colon V \to \{1, -1\}$ is defined as follows:
		
		\begin{itemize}
			\item $turn(v_0) \in \{1, -1\}$;
			
			\item $\forall (v, u) \in E, v \ne u \colon turn(v) := -turn(u)$.
		\end{itemize}
	\end{itemize}
\end{definition}

\subsection{Evaluation function and depth to winning}

\begin{definition}[$eval \colon V \to \{-1, 0, 1\}$]\label{def:eval_func}
	The evaluation function $eval$, which returns the outcome of a position supposing perfect play, is defined by induction as follows:
	
	\begin{flushleft}
		BASE
	\end{flushleft}
	
	$\forall v \in sinks(G) \colon eval(v) := win(v)$.
	
	\begin{flushleft}
		INDUCTION
	\end{flushleft}
	
	$\forall v \in V \setminus sinks(G)$:
	\begin{itemize}
		\item if $turn(v) = 1$:
		\[
		eval(v) := \max_{u \in G.Adj(v)} eval(u)
		\]
		
		\item if $turn(v) = -1$:
		\[
		eval(v) := \min_{u \in G.Adj(v)} eval(u)
		\]
	\end{itemize}
	
	As you can see, this is the classical definition that relies on minimax principle, which comes from game theory.
\end{definition}

\begin{definition}[Depth to winning]\label{def:win_depth}
	Let $g \in \{1, -1\}$. The depth to winning of $g$ is a function $d_g \colon V \to \N \cup \{ \infty \}$ defined by induction as follows.
	
	\begin{flushleft}
		BASE
	\end{flushleft}

	$\forall v \in sinks(G)$:
	\[
	d_g(v) :=
	\begin{cases}
		0, & \text{if } win(v) = g \\
		\infty, & \text{if } win(v) \ne g \\
	\end{cases}
	\]
	
	\begin{flushleft}
		INDUCTION
	\end{flushleft}
	
	$\forall v \in V \setminus sinks(G)$:
	\begin{itemize}
		\item if $turn(v) = g$:
		\[
		d_g(v) := \min_{u \in G.Adj(v)} d_g(u) + 1
		\]
		
		\item if $turn(v) \ne g$:
		\[
		d_g(v) := \max_{u \in G.Adj(v)} d_g(u) + 1
		\]
	\end{itemize}
\end{definition}

In the game of chess, this distance is commonly called \emph{depth to mate} (DTM). \pagebreak

\begin{lemma}[Good foundation of the evaluation function]\label{lemma:eval_well_founded}
	Let $g \in \{-1, 1\}$. For every $v \in V$:
	\[
	eval(v) = g \iff d_g(v) < \infty
	\]
\end{lemma}

\begin{proof}
	By induction on the depth to winning.
	
	\begin{flushleft}
		BASE: $v \in sinks(G)$
	\end{flushleft}
	
	By definitions \ref{def:eval_func} and \ref{def:win_depth}, $d_g(v) \ne \infty \iff g = win(v) = eval(v)$.
	
	\begin{flushleft}
		INDUCTION: $v \in V \setminus sinks(G)$
	\end{flushleft}
	
	\begin{flushleft}
		($\impliedby$)
	\end{flushleft}
	
	Suppose that $d_g(v) < \infty$.
	\begin{itemize}
		\item Case $turn(v) = g$.
		
		By def. \ref{def:win_depth}, $d_g(v) < \infty$ implies that $d_g(u) = d_g(v) - 1 < \infty$ for some $u \in G.Adj(v)$, so $eval(u) = g$ by induction, hence by def. \ref{def:eval_func}:
		\begin{itemize}
			\item if $turn(v) = 1$:
			\[
			eval(v) = \max_{u \in G.Adj(v)} eval(u) = 1
			\]
			
			\item if $turn(v) = -1$:
			\[
			eval(v) = \min_{u \in G.Adj(v)} eval(u) = -1
			\]
		\end{itemize}
	
		Therefore, $eval(v) = g$ anyway.
		
		\item Case $turn(v) \ne g$.
		
		By def. \ref{def:win_depth}, $d_g(v) < \infty$ implies that $d_g(u) \le d_g(v) - 1 < \infty$ for every $u \in G.Adj(v)$, so $eval(u) = g$ for every $u \in G.Adj(v)$ by induction, hence by def. \ref{def:eval_func}:
		\begin{itemize}
			\item if $turn(v) = -g = 1$:
			\[
			eval(v) = \max_{u \in G.Adj(v)} eval(u) = -1
			\]
			
			\item if $turn(v) = -g = -1$:
			\[
			eval(v) = \min_{u \in G.Adj(v)} eval(u) = 1
			\]
		\end{itemize}
		
		Therefore, $eval(v) = g$ anyway.
	\end{itemize}\pagebreak
	
	\begin{flushleft}
		($\implies$)
	\end{flushleft}
	
	Suppose that $eval(v) = g$.
	
	If $turn(v) = g$, by def. \ref{def:eval_func} and \ref{def:win_depth}:
	\[
	\begin{split}
		&eval(v) = g \\
		&\implies \exists u_1 \in G.Adj(v) \colon eval(u_1) = g \\
		&\implies \forall u_2 \in G.Adj(u_1) \colon eval(u_2) = g \\
		&\implies \exists u_3 \in G.Adj(u_2) \colon eval(u_3) = g \\
		&\implies \forall u_4 \in G.Adj(u_3) \colon eval(u_4) = g \\
	\end{split}
	\]
	and so on \dots
	
	Since the only way to get $eval(v) = g$ is to reach a base case in def. \ref{def:eval_func}, then this procedure must be finite. Hence, there must be a $k \in \N$ such that:
	\begin{equation}\label{eq:base_case_win}
		\begin{split}
			&\text{if $k$ is even:} \\
			&\exists u_1 \in G.Adj(v), \forall u_2 \in G.Adj(u_1), \dots, \exists u_{k-1} \in G.Adj(u_{k-2}), \forall u_k \in G.Adj(u_{k-1}) \colon win(u_k) = g \\
			&\text{if $k$ is odd:} \\
			&\exists u_1 \in G.Adj(v), \forall u_2 \in G.Adj(u_1), \dots, \forall u_{k-1} \in G.Adj(u_{k-2}), \exists u_k \in G.Adj(u_{k-1}) \colon win(u_k) = g
		\end{split}
	\end{equation}
	
	By def. \ref{def:win_depth}:
	\begin{equation}\label{eq:base_case_h}
		win(u_k) = g \implies d_g(u_k) = 0
	\end{equation}
	
	And, by induction and def. \ref{def:win_depth}, for every $n \in \{0, \dots, k-1\}$ we have:
	\[
	\begin{split}
		&\text{if $k-n$ is even:} \\
		&\forall u_{k-n} \in G.Adj(u_{k-n-1}) \colon d_g(u_{k-n}) \le n \implies d_g(u_{k-n-1}) \le n + 1 \\
		&\text{if $k-n$ is odd:} \\
		&\exists u_{k-n} \in G.Adj(u_{k-n-1}) \colon d_g(u_{k-n}) \le n \implies d_g(u_{k-n-1}) \le n + 1 \\
	\end{split}
	\]
	
	Therefore, by \eqref{eq:base_case_win} and \eqref{eq:base_case_h} we get $d_g(v) \le k$.
	
	For $turn(v) \ne g$, the reasoning is the same.
\end{proof} \pagebreak

\section{Solvability of the game model}\label{sect:solvalibility}

Clearly, the $eval$ function can be computed by an iterative dynamic programming algorithm that uses the backward induction technique, exploiting the inductive definition. However, we wonder whether this can be done using a recursive algorithm, rather than an iterative one.

Assuming the game has a finite number of legal positions, the answer seems to be yes, as Von Neumann theoretically proved in \cite[ch. 15]{von_neumann:game_theory}, but constructing such an algorithm may not be so straightforward. Starting from the initial position, the structure we get is a graph, not a tree. We know that the backtracking technique can be applied to trees, but we don't know whether the game graph actually represents a tree. This evidently has to do with the possibility of repeating the same position multiple times during a game (classic examples are chess and checkers), which is equivalent to the existence of cycles within the graph itself. A cycle is defined as a path in which the starting and ending vertices coincide (i.e., are indistinguishable).

In games, we usually ask ourselves what the winning strategy is, but here we also want to ask another question: is there a way to never lose? Note that there are three possible outcomes: win, lose, or draw. Therefore, not losing means drawing or winning. We will now write a recursive algorithm based on graph depth-first exploration (cf. \cite[22.3]{cormen:algorithms}) to build a winning strategy, and we will try to prove its correctness and completeness.

\textbf{Note:} the value of attributes not yet initialized is assumed to be \textbf{NIL}.

\begin{algorithmic}[1]
	\Procedure{WinTree}{$G, g$}
		\State \Call{WinTreeRecursion}{$G, g, G.v_0$}
	\EndProcedure
	
	\Procedure{WinTreeRecursion}{$G, g, v$}
		\State $v.color \gets \mathrm{GRAY}$ \label{line:v_gray}
		
		\If{$v \in sinks(G)$} \Comment{base of recursion}
			\If{$win(v) = g$}\label{line:base_case}
				\State $v.win \gets \mathrm{TRUE}$
				\State $v.d \gets 0$
			\Else
				\State $v.win \gets \mathrm{FALSE}$
				\State $v.d \gets \infty$
			\EndIf \pagebreak
			
		\Else \Comment{recursive step}
			\If{$turn(v) = g$}\label{line:init_win}
				\State $v.win \gets \mathrm{FALSE}$ 
				\State $v.d \gets \infty$
			\Else
				\State $v.win \gets \mathrm{TRUE}$
				\State $v.d \gets 0$
			\EndIf
			
			\State $end \gets \mathrm{FALSE}$
			\State $i \gets 0$
			\While{\textbf{not} $end$ \textbf{and} $i < |G.Adj(v)|$}
				\State $w \gets G.Adj(v)[i]$
				
				\If{$w.color = \mathrm{GRAY}$}\label{line:if_gray}
					\State $win \gets \textbf{NIL}$ \Comment{The winning tree can't have cycles}
					\State $d \gets \infty$
				\Else\label{line:if_not_gray}
					\If{$w.color \ne \mathrm{BLACK}$}\label{line:recursion}
						\State \Call{WinTreeRecursion}{$G, g, w$}
					\EndIf
				
					\State $win \gets w.win$\label{line:rec_win_assign}
					\State $d \gets w.d + 1$
				\EndIf
				
				\If{$win \ne (turn(v) \ne g)$}\label{line:change_win}
					\State $v.win \gets win$
					\State $v.d \gets d$
					\If{$v.win \ne \textbf{NIL}$}
						\State $end \gets \mathrm{TRUE}$
					\EndIf
				\ElsIf{$turn(v) \ne g$ \textbf{and} $d > v.d$}\label{line:depth_else}
					\State $v.d \gets d$
				\EndIf
				
				\State $i \gets i + 1$
			\EndWhile
		\EndIf
		
		\If{$v.win = \textbf{NIL}$}\label{line:color_if}
			\State $v.color \gets \mathrm{WHITE}$
		\Else
			\State $v.color \gets \mathrm{BLACK}$
		\EndIf
	\EndProcedure
\end{algorithmic}

It's easy to see that this algorithm terminates. Indeed, although a vertex could be visited multiple times, marking it as gray ensures that it can only be visited once before its branch of the depth-first tree (cf. \cite[22.3]{cormen:algorithms}) is traversed again. In this way, the branches can only be reached through simple paths from the root vertex, and since each outgoing arc of the root cannot be traversed more than once (because the root cannot be visited again until the entire tree has been traversed), then each simple path starting from the root is never traversed more than once, hence each branch can only be traversed a finite number of times. Therefore, the time complexity is $O(d^l \cdot d \cdot l)$, where $d = \max\{|G.Adj(v)| \mid v \in V\}$ is the maximum degree of the vertices and $l = \max\{n \mid \left< v_1, \dots, v_n \right> \text{ is a simple path in } G\}$ is the maximum length of a simple path.

\subsection{Correctness and completeness of the algorithm}

\begin{theorem}[Correctness of the winning strategy]\label{theo:correctness}
	Let $g \in \{1, -1\}$. After the execution of the procedure $\Call{WinTree}{G, g}$, for every $v \in V$ we have:
	\[
	v.win = \mathrm{TRUE} \implies eval(v) = g
	\]
\end{theorem}

\begin{proof}
	By structural induction on the recursion tree.
	
	Suppose that $v.win = \mathrm{TRUE}$.
	
	\begin{flushleft}
		BASE: $v \in sinks(G)$
	\end{flushleft}
	
	 $v.win = \mathrm{TRUE}$ implies that an invocation to $\Call{WinTreeRecursion}{G, g, v}$ was made, and then that $win(v) = g$ by the if at line \ref{line:base_case}, hence $eval(v) = g$ by def. \ref{def:eval_func}.
	
	\begin{flushleft}
		INDUCTION: $v \in V \setminus sinks(G)$
	\end{flushleft}

	\begin{itemize}
		\item If $turn(v) = g$:
		
		$v.win$ initially gets the value $\mathrm{FALSE}$ by the if at line \ref{line:init_win}, so, by the else at line \ref{line:if_not_gray} and the if at line \ref{line:change_win}, $v.win = \mathrm{TRUE}$ implies that $w.win = \mathrm{TRUE}$ for some $w \in G.Adj(v)$. The if at line \ref{line:recursion} also implies that we can apply the inductive hypothesis on $w$ (because it ensures that there has been an invocation to $\Call{WinTreeRecursion}{G, g, w}$ that caused $w.win = \mathrm{TRUE}$), so $w.win = \mathrm{TRUE}$ implies $eval(w) = g$, hence $eval(v) = g$ by def. \ref{def:eval_func}.
		
		\item If $turn(v) \ne g$:
		$v.win$ initially gets the value $\mathrm{TRUE}$ by the if at line \ref{line:init_win}, so, by the else at line \ref{line:if_not_gray} and the if at line \ref{line:change_win}, $v.win = \mathrm{TRUE}$ implies that $w.win = \mathrm{TRUE}$ for every $w \in G.Adj(v)$. The if at line \ref{line:recursion} also implies that we can apply the inductive hypothesis on $w$, so $w.win = \mathrm{TRUE}$ implies $eval(w) = g$ for every $w \in G.Adj(v)$, hence $eval(v) = g$ by def. \ref{def:eval_func}.\qedhere
	\end{itemize}
\end{proof}

\begin{theorem}[Completeness for black vertices]\label{theo:compl_black_vert}
	Let $g \in \{1, -1\}$. After the execution of the procedure $\Call{WinTree}{G, g}$, for every $v \in V$ such that $v.color = \mathrm{BLACK}$ we have:
	\[
	 eval(v) = g \implies v.win = \mathrm{TRUE}
	\]
\end{theorem}\pagebreak

\begin{proof}
	By induction on the depth to winning.
	
	Suppose that $eval(v) = g$.
	
	\begin{flushleft}
		BASE: $v \in sinks(G)$
	\end{flushleft}
	
	$v.color = \mathrm{BLACK}$ implies that an invocation to $\Call{WinTreeRecursion}{G, g, v}$ was made, and since $eval(v) = g$ implies $win(v) = g$ by def. \ref{def:eval_func}, then $v.win = \mathrm{TRUE}$ by the if at line \ref{line:base_case}.
	
	\begin{flushleft}
		INDUCTION: $v \in V \setminus sinks(G)$
	\end{flushleft}
	
	By the if at line \ref{line:color_if}, $v.color = \mathrm{BLACK}$ implies $v.win \ne \textbf{NIL}$, so it can only be that $v.win = \mathrm{TRUE}$ or $v.win = \mathrm{FALSE}$.
	
	Suppose by contradiction that $v.win = \mathrm{FALSE}$.
	\begin{itemize}
		\item If $turn(v) = g$:
		
		$v.win$ initially gets the value $\mathrm{FALSE}$ by the if at line \ref{line:init_win}, so $v.win = \mathrm{FALSE}$ implies that $u.win = \mathrm{FALSE}$ for all $u \in G.Adj(v)$ by the if at line \ref{line:change_win} and the else at line \ref{line:if_not_gray}. Furthermore, by the if at line \ref{line:color_if}, $u.win = \mathrm{FALSE}$ implies $u.color = \mathrm{BLACK}$ for all $u \in G.Adj(v)$.
		
		Since $eval(v) = g$, then $eval(w) = g$ for some $w \in G.Adj(v)$ by def. \ref{def:eval_func}. Without loss of generality, we can assume that
		\[
		d_g(w) = \min_{t \in G.Adj(v) \colon eval(t) = g} d_g(t)
		\]
		
		Since, by the lemma \ref{lemma:eval_well_founded}, $eval(w) = g$ implies $d_g(w) < \infty$, and since $d_g(u) = \infty$ for every $u \in G.Adj(v)$ such that $eval(u) \ne g$, then $d_g(w) = d_g(v) - 1 < \infty$ by def. \ref{def:win_depth}. This means that we can apply the inductive hypothesis on $w$, so $w.color = \mathrm{BLACK}$ and $eval(w) = g$ imply $w.win = \mathrm{TRUE}$, contradicting $u.win = \mathrm{FALSE}$ for all $u \in G.Adj(v)$. Therefore, it must be that $v.win = \mathrm{TRUE}$.
		
		\item If $turn(v) \ne g$:
		
		$v.win$ initially gets the value $\mathrm{TRUE}$ by the if at line \ref{line:init_win}, so $v.win = \mathrm{FALSE}$ implies that $u.win = \mathrm{FALSE}$ for some $u \in G.Adj(v)$ by the if at line \ref{line:change_win} and the else at line \ref{line:if_not_gray}. Furthermore, by the if at line \ref{line:color_if}, $u.win = \mathrm{FALSE}$ implies $u.color = \mathrm{BLACK}$.
		
		Since $eval(v) = g$, then $eval(w) = g$ for all $w \in G.Adj(v)$ by def. \ref{def:eval_func}. By the lemma \ref{lemma:eval_well_founded}, $eval(w) = g$ implies $d_g(w) < \infty$, hence we have by def. \ref{def:win_depth}:
		\[
		d_g(w) \le \max_{t \in G.Adj(v)} d_g(t) = d_g(v) - 1 < \infty
		\]
		
		This means that we can apply the inductive hypothesis on $u$, so $u.color = \mathrm{BLACK}$ and $eval(u) = g$ imply $u.win = \mathrm{TRUE}$, contradicting $u.win = \mathrm{FALSE}$. Therefore, it must be that $v.win = \mathrm{TRUE}$.\qedhere
	\end{itemize}
\end{proof}\pagebreak

\begin{corollary}[Complementary correctness]\label{corol:compl_correctness}
	Let $g \in \{1, -1\}$. After the execution of the procedure $\Call{WinTree}{G, g}$, for every $v \in V$ we have:
	\[
	v.win = \mathrm{FALSE} \implies eval(v) \ne g
	\]
\end{corollary}

\begin{proof}
	By the if at line \ref{line:color_if}, $v.win = \mathrm{FALSE}$ implies $v.color = \mathrm{BLACK}$, therefore $v.win = \mathrm{FALSE}$ entails $eval(v) \ne g$ by the theorem \ref{theo:compl_black_vert}.
\end{proof}

\begin{theorem}[Completeness of the search]\label{theo:completeness}
	Let $g \in \{1, -1\}$. After the execution of the procedure $\Call{WinTree}{G, g}$, we have:
	\[
	eval(v_0) = g \implies v_0.win = \mathrm{TRUE}
	\]
\end{theorem}

\begin{proof}
	By infinite descent.
	
	Suppose that $eval(v_0) = g$. If we prove that $v_0.color = \mathrm{BLACK}$, then $v_0.win = \mathrm{TRUE}$ follows by the theorem \ref{theo:compl_black_vert}. So, let's prove that.
	
	Suppose by contradiction that $v_0.color \ne \mathrm{BLACK}$ after the invocation of $\Call{WinTree}{G, g}$. Then $v_0.win = \textbf{NIL}$ by the if at line \ref{line:color_if}.
	
	If $turn(v_0) = g$:
	\begin{enumerate}
		\item $eval(v_0) = g$ implies $eval(v_1) = g$ for some $v_1 \in G.Adj(v_0)$ by def. \ref{def:eval_func}. Since $v_0.win = \textbf{NIL}$ implies $v_0 \notin sinks(G)$ by the if at line \ref{line:base_case}, we can assume that $d_g(v_1) = d_g(v_0) - 1 < \infty$ without loss of generality, by the same reasoning of theorem \ref{theo:compl_black_vert} on the depths.
		
		Since $v_0.win$ initially gets the value $\mathrm{FALSE}$ by the if at line \ref{line:init_win}, $v_0.win = \textbf{NIL}$ implies $v_1.win \ne \mathrm{TRUE}$ or $v_1.color = \mathrm{GRAY}$ by the ifs at lines \ref{line:change_win}, \ref{line:if_gray} and the else at line \ref{line:if_not_gray}. Since the set of gray vertices is $\{v_0\}$ at the invocation of $\Call{WinTreeRecursion}{G, g, v_0}$ and $d_g(v_1) < d_g(v_0)$, then $v_1.color \ne \mathrm{GRAY}$ and so $v_1.win \ne \mathrm{TRUE}$. Hence, by the theorem \ref{theo:compl_black_vert}, $eval(v_1) = g$ implies $v_1.color \ne \mathrm{BLACK}$, so the invocation of $\Call{WinTreeRecursion}{G, g, v_1}$ was made by the if at line \ref{line:recursion} and $v_1.win = \textbf{NIL}$ by the if at line \ref{line:color_if}.
		
		\item Since $v_1.win$ initially gets the value $\mathrm{TRUE}$ by the if at line \ref{line:init_win}, $v_1.win = \textbf{NIL}$ implies $v_2.win = \textbf{NIL}$ or $v_2.color = \mathrm{GRAY}$ for some $v_2 \in G.Adj(v_1)$  by the ifs at lines \ref{line:change_win}, \ref{line:if_gray} and the else at line \ref{line:if_not_gray}. Furthermore, we have $d_g(v_2) \le d_g(v_1) - 1 < \infty$ by the same reasoning of theorem \ref{theo:compl_black_vert} on the depths.
		
		Since the set of gray vertices is $\{v_0, v_1\}$ at the invocation of $\Call{WinTreeRecursion}{G, g, v_1}$ and $d_g(v_2) < d_g(v_1) < d_g(v_0)$, then $v_2.color \ne \mathrm{GRAY}$ and so $v_2.win = \textbf{NIL}$, thus $v_2.color \ne \mathrm{BLACK}$ by the if at line \ref{line:color_if}, hence the invocation of $\Call{WinTreeRecursion}{G, g, v_2}$ was made by the if at line \ref{line:recursion}. Furthermore, $turn(v_1) \ne g$ and $eval(v_1) = g$ imply $eval(v_2) = g$ by def. \ref{def:eval_func}.
	\end{enumerate}
	
	Applying the reasoning at point 1 with $v_2$ instead of $v_0$, we have $d_g(v_3) < d_g(v_2) < d_g(v_1) < d_g(v_0)$ for some $v_3 \in G.Adj(v_2)$, and applying the reasoning at point 2 with $v_3$ instead of $v_1$, we have $d_g(v_4) < d_g(v_3) < d_g(v_2) < d_g(v_1) < d_g(v_0)$ for some $v_4 \in G.Adj(v_3)$, and so on \dots
	
	Going forward with this procedure, we get an infinite succession of vertices $\{v_k\}_{k \in \N}$, where $d_g(v_{k+1}) < d_g(v_k)$. But $eval(v_0) = g$ implies that $d_g(v_0) < \infty$ by the lemma \ref{lemma:eval_well_founded}, therefore, as $d_g(v_k) \in \N$ by def. \ref{def:win_depth}, the succession cannot be infinite, leading to a contradiction.
	
	For $turn(v_0) \ne g$ the reasoning is the same, but the order of application of the points 1 and 2 is reversed.
\end{proof}

\subsection{Estimate of the depth to winning}

\begin{theorem}\label{theo:finite_depth_win_equiv}
	Let $g \in \{1, -1\}$. After the execution of the procedure $\Call{WinTree}{G, g}$, for every $v \in V$ such that $v.color \ne \mathbf{NIL}$ we have:
	\[
	v.win = \mathrm{TRUE} \iff v.d < \infty
	\]
\end{theorem}

\begin{proof}
	By structural induction on the recursion tree.
	
	\begin{flushleft}
		BASE: $v \in sinks(G)$
	\end{flushleft}
	
	It follows from the if at line \ref{line:base_case}.
	
	\begin{flushleft}
		INDUCTION: $v \in V \setminus sinks(G)$
	\end{flushleft}
	
	\begin{flushleft}
		($\implies$)
	\end{flushleft}
	
	Suppose that $v.win = \mathrm{TRUE}$.
	\begin{itemize}
		\item If $turn(v) = g$:
		
		$v.win$ initially gets the value $\mathrm{FALSE}$ by the if at line \ref{line:init_win}, so $v.win = \mathrm{TRUE}$ implies that $w.win = \mathrm{TRUE}$ and $v.d = w.d + 1$ for some $w \in G.Adj(v)$ by the if at line \ref{line:change_win} and the else at line \ref{line:if_not_gray}. This means that $w.color \ne \textbf{NIL}$ and so we can apply the inductive hypothesis on $w$, hence $w.win = \mathrm{TRUE}$ implies $w.d < \infty$, therefore $v.d = w.d + 1$ implies that $v.d < \infty$.
		
		\item If $turn(v) \ne g$:
		
		$v.win$ initially gets the value $\mathrm{TRUE}$ and $v.d$ the value 0 by the if at line \ref{line:init_win}, so $v.win = \mathrm{TRUE}$ implies that $w.win = \mathrm{TRUE}$ and $v.d = w.d + 1$ for some $w \in G.Adj(v)$ by the elses at lines \ref{line:depth_else} and \ref{line:if_not_gray}. This means that $w.color \ne \textbf{NIL}$ and so we can apply the inductive hypothesis on $w$, hence $w.win = \mathrm{TRUE}$ implies $w.d < \infty$, therefore $v.d = w.d + 1$ implies that $v.d < \infty$.
	\end{itemize}
	
	\begin{flushleft}
		($\impliedby$)
	\end{flushleft}
	
	Suppose that $v.d < \infty$.
	\begin{itemize}
		\item If $turn(v) = g$:
		
		$v.d$ initially gets the value $\infty$ by the if at line \ref{line:init_win}, so $v.d < \infty$ implies that $w.d = v.d - 1 < \infty$ for some $w \in G.Adj(v)$ by the if at line \ref{line:change_win} and the else at line \ref{line:if_not_gray}. This means that $w.color \ne \textbf{NIL}$ and so we can apply the inductive hypothesis on $w$, hence $w.d < \infty$ implies $w.win = \mathrm{TRUE}$, therefore $v.win = \mathrm{TRUE}$ by the else at line \ref{line:if_not_gray} and the if at line \ref{line:change_win}.\pagebreak
		
		\item If $turn(v) \ne g$:
		
		$v.d$ initially gets the value 0 by the if at line \ref{line:init_win}, so $v.d < \infty$ implies that $w.d + 1 \le v.d < \infty$ for all $w \in G.Adj(v)$ by the elses at lines \ref{line:depth_else} and \ref{line:if_not_gray}. This means that $w.color \ne \textbf{NIL}$ and so we can apply the inductive hypothesis on $w$, hence $w.d < \infty$ implies $w.win = \mathrm{TRUE}$ for all $w \in G.Adj(v)$, therefore $v.win = \mathrm{TRUE}$ by the else at line \ref{line:if_not_gray} and the if at line \ref{line:change_win}.\qedhere
	\end{itemize}
\end{proof}

\begin{theorem}\label{theo:depth_upper_bound}
	Let $g \in \{1, -1\}$. After the execution of the procedure $\Call{WinTree}{G, g}$, for every $v \in V$ such that $v.color \ne \mathbf{NIL}$ we have $v.d \ge d_g(v)$.
\end{theorem}

\begin{proof}
	By structural induction on the recursion tree.
	
	\begin{flushleft}
		BASE: $v \in sinks(G)$
	\end{flushleft}
	
	By the if at line \ref{line:base_case} and by def. \ref{def:win_depth}, $v.d = d_g(v)$.
	
	\begin{flushleft}
		INDUCTION: $v \in V \setminus sinks(G)$
	\end{flushleft}
	
	By the if at line \ref{line:color_if}, $v.color \ne \mathbf{NIL}$ implies $v.color = \mathrm{WHITE}$ or $v.color = \mathrm{BLACK}$.
	
	\begin{itemize}
		\item If $v.color = \mathrm{WHITE}$:
		
		$v.win = \mathbf{NIL}$ by the if at line \ref{line:color_if}, so $v.d = \infty$ by the theorem \ref{theo:finite_depth_win_equiv}, hence $v.d \ge d_g(v)$.
		
		\item If $v.color = \mathrm{BLACK}$:
		
		$v.win \ne \mathbf{NIL}$ by the if at line \ref{line:color_if}, so it can only be that $v.win = \mathrm{TRUE}$ or $v.win = \mathrm{FALSE}$.
		\begin{itemize}
			\item If $v.win = \mathrm{FALSE}$, then $v.d = \infty$ by the theorem \ref{theo:finite_depth_win_equiv}, hence $v.d \ge d_g(v)$.
			
			\item If $v.win = \mathrm{TRUE}$:
			\begin{itemize}
				\item If $turn(v) = g$:
				
				$v.win$ initially gets the value $\mathrm{FALSE}$ by the if at line \ref{line:init_win}, so $v.win = \mathrm{TRUE}$ implies that $w.win = \mathrm{TRUE}$ and $v.d = w.d + 1$ for some $w \in G.Adj(v)$ by the if at line \ref{line:change_win} and the else at line \ref{line:if_not_gray}. This means that $w.color \ne \textbf{NIL}$ and so we can apply the inductive hypothesis on $w$, hence $d_g(w) \le w.d$. Since $d_g(v) \le d_g(w) + 1$ by def. \ref{def:win_depth}, then $d_g(v) \le d_g(w) + 1 \le w.d + 1 = v.d$.
				
				\item If $turn(v) \ne g$:
				
				$v.win$ initially gets the value $\mathrm{TRUE}$ and $v.d$ the value 0 by the if at line \ref{line:init_win}, so $v.win = \mathrm{TRUE}$ implies that $w.win = \mathrm{TRUE}$ and $w.d + 1 \le v.d$ for all $w \in G.Adj(v)$ by the elses at lines \ref{line:depth_else} and \ref{line:if_not_gray}. This means that $w.color \ne \textbf{NIL}$ and so we can apply the inductive hypothesis on $w$, hence $d_g(w) \le w.d$ for all $w \in G.Adj(v)$.
				
				Since $d_g(v) = d_g(u) + 1$ for some $u \in G.Adj(v)$ by def. \ref{def:win_depth}, and since $d_g(u) \le u.d$ and $u.d + 1 \le v.d$, then $d_g(v) = d_g(u) + 1 \le u.d + 1 \le v.d$. \qedhere
			\end{itemize}
		\end{itemize}
	\end{itemize}
\end{proof}

Note that the theorem \ref{theo:finite_depth_win_equiv} is somewhat of a \emph{computational version} of the lemma \ref{lemma:eval_well_founded}. Indeed, we have the corollaries below that follow easily from the previous results of this article, whose proof is left to the reader. Also note that using theorems \ref{theo:finite_depth_win_equiv}, \ref{theo:depth_upper_bound} in conjunction with the lemma \ref{lemma:eval_well_founded} leads to a further proof of the correctness theorem \ref{theo:correctness}.

\begin{corollary}\label{corol:finite_depth_eval_equiv}
	Let $g \in \{1, -1\}$. After the execution of the procedure $\Call{WinTree}{G, g}$, for every $v \in V$ such that $v.color = \mathrm{BLACK}$ we have:
	\[
	v.d < \infty \iff eval(v) = g
	\]
\end{corollary}

\begin{corollary}
	Let $g \in \{1, -1\}$. After the execution of the procedure $\Call{WinTree}{G, g}$, for every $v \in V$ such that $v.color = \mathrm{BLACK}$ we have:
	\[
	v.d < \infty \iff d_g(v) < \infty 
	\]
\end{corollary} \pagebreak

\section{Implementation strategies and optimizations}

\subsection{Order of visiting adjacent vertices}

A first simple idea for reducing useless recursive invocations is to iterate first on the black vertices (that is, those that do not require a recursive invocation, since the value of their attributes has already been determined) and use it for the algorithm seen in the section \ref{sect:solvalibility}. Therefore, the algorithm becomes as follows:

\begin{algorithmic}[1]
	\Procedure{WinTreeRecursion}{$G, g, v$}
		\State $v.color \gets \mathrm{GRAY}$
		
		\If{$v \in sinks(G)$} \Comment{base of recursion}
			\If{$win(v) = g$}
				\State $v.win \gets \mathrm{TRUE}$
				\State $v.d \gets 0$
			\Else
				\State $v.win \gets \mathrm{FALSE}$
				\State $v.d \gets \infty$
			\EndIf
		\Else \Comment{recursive step}
			\If{$turn(v) = g$}
				\State $v.win \gets \mathrm{FALSE}$ 
				\State $v.d \gets \infty$
			\Else
				\State $v.win \gets \mathrm{TRUE}$
				\State $v.d \gets 0$
			\EndIf
			
			\State $end \gets \mathrm{FALSE}$
			\State $i \gets 0$
			\While{\textbf{not} $end$ \textbf{and} $i < |G.Adj(v)|$}
				\State $w \gets G.Adj(v)[i]$
				
				\If{$w.color = \mathrm{BLACK}$}
					\If{$w.win \ne (turn(v) \ne g)$}
						\State $v.win \gets w.win$
						\State $v.d \gets w.d + 1$
						\State $end \gets \mathrm{TRUE}$
					\ElsIf{$turn(v) \ne g$ \textbf{and} $w.d + 1 > v.d$}
						\State $v.d \gets w.d + 1$
					\EndIf
				\EndIf
				
				\State $i \gets i + 1$
			\EndWhile \pagebreak
			
			\State $i \gets 0$
			\While{\textbf{not} $end$ \textbf{and} $i < |G.Adj(v)|$}
				\State $w \gets G.Adj(v)[i]$
				
				\If{$w.color \ne \mathrm{BLACK}$}
					\If{$w.color = \mathrm{GRAY}$}
						\State $win \gets \textbf{NIL}$
						\State $d \gets \infty$
					\Else \Comment{$w.color \ne \mathrm{BLACK}$}
						\State \Call{WinTreeRecursion}{$G, g, w$}
						\State $win \gets w.win$
						\State $d \gets w.d + 1$
					\EndIf
					
					\If{$win \ne (turn(v) \ne g)$}
						\State $v.win \gets win$
						\State $v.d \gets d$
						\If{$v.win \ne \textbf{NIL}$}
							\State $end \gets \mathrm{TRUE}$
						\EndIf
					\ElsIf{$turn(v) \ne g$ \textbf{and} $d > v.d$}
						\State $v.d \gets d$
					\EndIf
				\EndIf
				
				\State $i \gets i + 1$
			\EndWhile
		\EndIf
		
		\If{$v.win = \textbf{NIL}$}
			\State $v.color \gets \mathrm{WHITE}$
		\Else
			\State $v.color \gets \mathrm{BLACK}$
		\EndIf
	\EndProcedure
\end{algorithmic}

Regarding the order of visiting non-black vertices, if we're talking about chess, for example, it's preferable to first analyze what chess engines consider the "strongest" moves. In general, if programs exist that provide heuristic evaluations of positions, it's useful to use them to analyze moves in order of playing strength, minimizing the analysis of weaker moves, so as to terminate the iteration on the vertices as quickly as possible. This approach is a sort of hybrid between the type A and type B strategies described in Shannon's 1950 work \cite{shannon:computer_playing_chess}. \pagebreak

\subsection{How to determine a draw strategy}

Suppose we want to find a way for player $g \in\{1, -1\}$ to draw, starting at the vertex $v$ with $turn(v) = g$. To do this, we could determine if there is no way for $-g$ to win, by invoking $\Call{WinTreeRecursion}{G, -g, v}$. If we get $v.win \ne \mathrm{TRUE}$ after the execution, then $eval(v) \ne -g$ by the completeness theorem \ref{theo:completeness}, and this means there is a way for $g$ at least to draw. But obviously this is not enough, we would also like to know which of the vertices adjacent to $v$ leads to a draw (or possibly to a win), that is, what move we should make to draw.

It could be that $w.win = \mathrm{FALSE}$ for some $w \in G.Adj(v)$. In that case, the corollary \ref{corol:compl_correctness} assures that $eval(w) \ne -g$, so $w$ is the correct choice. However, we might not be so lucky. Indeed, if $u.win \ne \mathrm{FALSE}$ for all $u \in G.Adj(v)$, the theorem \ref{theo:completeness} only assures that $eval(w) \ne -g$ for some $w \in G.Adj(v)$, but there are no other results we can exploit to know which one it is: there's no theorem that assures $w.win = \textbf{NIL} \implies eval(w) \ne -g$ if $w$ is a white vertex. The only exception is when there is only one $w \in G.Adj(v)$ such that $w.win = \textbf{NIL}$, in that case the choice is obvious.

Therefore, it is clear that the counter-nominal version of the completeness theorem is not constructive. To overcome this problem, we need to prove a stronger completeness theorem by formally defining the concept of winning strategy.

\begin{definition}[Tree terminology]
	Given a tree $T$, we will use the following terminology.
	\begin{itemize}
		\item $T.root$ is the root of $T$;
		
		\item $T.n$ is the number of child subtrees of $T$;
		
		\item $T.c_1, \dots, T.c_{T.n}$ are the child subtrees of $T$.
	\end{itemize}
\end{definition}

\begin{definition}[Height of a tree]
	Given a tree $T$, the height $h(T)$ is defined by induction as follows.
	\begin{itemize}
		\item if $T.n = 0$, then $h(T) := 0$;
		
		\item if $T.n \ge 1$, then $h(T) := \max\{h(T.c_i) \mid i \in \{1, \dots, T.n\}\} + 1$.
	\end{itemize}
\end{definition}

\begin{definition}[Nodes of a tree]
	Given a tree $T$, the set $nodes(T)$ is defined by induction as follows:
	\[
	nodes(T) := \{T.root\} \cup \bigcup_{i = 1}^{T.n} nodes(T.c_i)
	\]
\end{definition} \pagebreak

\begin{definition}[Winning trees]\label{def:win_tree}
	Let $g \in \{1, -1\}$. The set of trees $WinTrees(G, g)$ is defined by induction as follows.
	
	Let $T$ be a tree such that $T.root \in V$.
	
	\begin{flushleft}
		BASE: $T.root \in sinks(G)$
	\end{flushleft}
	\[
	T \in WinTrees(G, g) \iff win(T.root) = g \text{ and } T.n = 0
	\]
	
	\begin{flushleft}
		INDUCTION: $T.root \in V \setminus sinks(G)$
	\end{flushleft}

	\begin{itemize}
		\item If $turn(T.root) = g$:
		
		$T \in WinTrees(G, g)$ if and only if:
		\begin{itemize}
			\item $T.n = 1$;
			
			\item $T.c_1 \in WinTrees(G, g)$;
			
			\item $T.c_1.root \in G.Adj(T.root)$.
		\end{itemize}
		
		\item If $turn(T.root) \ne g$:
		
		$T \in WinTrees(G, g)$ if and only if:
		\begin{itemize}
			\item $T.n = |G.Adj(T.root)|$;
			
			\item  for all $i \in \{1, \dots, T.n\}$: $T.c_i \in WinTrees(G, g)$;
			
			\item for all $w \in G.Adj(T.root)$, there exists $i_w \in \{1, \dots, T.n\}$ such that $w = T.c_{i_w}.root$.
		\end{itemize}
	\end{itemize}
\end{definition}

\begin{lemma}[Correctness of winning trees]\label{lemma:correctness_trees}
	If $g \in \{1, -1\}$ and $T \in WinTrees(G, g)$, then $eval(T.root) = g$.
\end{lemma}

\begin{proof}
	By structural induction on $T$.
	
	\begin{flushleft}
		BASE: $T.root \in sinks(G)$
	\end{flushleft}
	
	$T \in WinTrees(G, g)$ implies $win(T.root) = g$, therefore $eval(T.root) = g$ by def. \ref{def:eval_func}.
	
	\begin{flushleft}
		INDUCTION: $T.root \in V \setminus sinks(G)$
	\end{flushleft}
	
	\begin{itemize}
		\item If $turn(T.root) = g$:
		
		by def. \ref{def:win_tree}, $T.c_1 \in WinTrees(G, g)$ and $T.c_1.root \in G.Adj(T.root)$, so $eval(T.c_1.root) = g$ by induction, therefore $eval(T.root) = g$ by def. \ref{def:eval_func}.
		
		\item If $turn(T.root) \ne g$:
		
		by def. \ref{def:win_tree}, for all $w \in G.Adj(T.root)$ there exists $i_w \in \{1, \dots, T.n\}$ such that $w = T.c_{i_w}.root$ and $T.c_{i_w} \in WinTrees(G, g)$, so $eval(w) = g$ by induction, therefore $eval(T.root) = g$ by def. \ref{def:eval_func}.\qedhere
	\end{itemize}
\end{proof}\pagebreak

\begin{lemma}[Completeness of winning trees]\label{lemma:completeness_trees}
	Let $g \in \{1, -1\}$ and $v \in G.V$ such that $eval(v) = g$. Then, there exists $T \in WinTrees(G, g)$ such that $T.root = v$.
\end{lemma}

\begin{proof}
	By induction on the depth to winning.
	
	Let $T.root := v$.
	
	\begin{flushleft}
		BASE: $v \in sinks(G)$
	\end{flushleft}
	
	$eval(v) = g$ implies $win(v) = g$, so letting $T.n := 0$ we get $T \in WinTrees(G, g)$.
	
	\begin{flushleft}
		INDUCTION: $v \in V \setminus sinks(G)$
	\end{flushleft}
	
	\begin{itemize}
		\item If $turn(v) = g$:
		
		since $eval(v) = g$, then $eval(w) = g$ for some $w \in G.Adj(v)$ by def. \ref{def:eval_func}. Without loss of generality, we can assume that
		\[
		d_g(w) = \min_{t \in G.Adj(v) \colon eval(t) = g} d_g(t)
		\]
		
		Since, by the lemma \ref{lemma:eval_well_founded}, $eval(w) = g$ implies $d_g(w) < \infty$, and since $d_g(u) = \infty$ for every $u \in G.Adj(v)$ such that $eval(u) \ne g$, then $d_g(w) = d_g(v) - 1 < \infty$ by def. \ref{def:win_depth}. This means we can apply the inductive hypothesis on $w$, so $eval(w) = g$ implies that there is a $T' \in WinTrees(G, g)$ such that $T'.root = w$, therefore, letting $T.n := 1$ and $T.c_1 := T'$, we get $T \in WinTrees(G, g)$.
		
		\item If $turn(v) \ne g$:
		
		Let $\{w_1, \dots, w_{|G.Adj(v)|}\} := G.Adj(v)$. Since $eval(v) = g$, then $eval(w_i) = g$ for all $i \in \{1, \dots, |G.Adj(v)|\}$ by def. \ref{def:eval_func}. By the lemma \ref{lemma:eval_well_founded}, $eval(w_i) = g$ implies $d_g(w_i) < \infty$, hence we have by def. \ref{def:win_depth}:
		\[
		d_g(w_i) \le \max_{t \in G.Adj(v)} d_g(t) = d_g(v) - 1 < \infty
		\]
		
		This means we can apply the inductive hypothesis on $w_i$, so $eval(w_i) = g$ implies that there is a $T'_i \in WinTrees(G, g)$ such that $T'_i.root = w_i$, therefore, letting $T.n := |G.Adj(v)|$ and $T.c_i := T'_i$ for all $i \in \{1, \dots, |G.Adj(v)|\}$, we get $T \in WinTrees(G, g)$.\qedhere
	\end{itemize}
\end{proof}\pagebreak

\begin{lemma}\label{lemma:correct_tree_construction}
	Let $g \in \{1, -1\}$, $v \in G.V$ and $S := \{w \in G.V \mid w.color = \mathrm{GRAY} \text{ and } v \leadsto w\}$ before calling $\Call{WinTreeRecursion}{G, g, v}$.
	
	If there exists $T \in WinTrees(G, g)$ such that $T.root = v$ and $nodes(T) \cap S = \emptyset$, then $v.win = \mathrm{TRUE}$ after the execution of $\Call{WinTreeRecursion}{G, g, v}$.
\end{lemma}

\begin{proof}
	By structural induction on the recursion tree.
	
	Suppose there exists a $T \in WinTrees(G, g)$ such that $T.root = v$ and $nodes(T) \cap S = \emptyset$.
	
	\begin{flushleft}
		BASE: $v \in sinks(G)$
	\end{flushleft}
	
	By def. \ref{def:win_tree}, the hypothesis implies $win(v) = g$, so $v.win = \mathrm{TRUE}$ by the if at line \ref{line:base_case}.
	
	\begin{flushleft}
		INDUCTION: $v \in V \setminus sinks(G)$
	\end{flushleft}

	Without loss of generality, we can assume that the height of $T$ is minimal, that is:
	\begin{equation}\label{eq:min_height_T}
		h(T) = \min\{h(X) \mid X \in WinTrees(G, g),\ X.root = v,\ nodes(X) \cap S = \emptyset\}
	\end{equation}
	
	\begin{itemize}
		\item If $turn(v) = g$:
		
		by def. \ref{def:win_tree}, the hypothesis implies $T.c_1 \in WinTrees(G, g)$ and $T.c_1.root \in G.Adj(v)$. Since the nodes of $T$ are not in $S$, then $nodes(T.c_1) \cap S = \emptyset$ and $w.color \ne \mathrm{GRAY}$ at the execution of line \ref{line:if_gray} with $w := T.c_1.root$.
		
		Let $S' := \{u \in G.V \mid u.color = \mathrm{GRAY} \text{ and } w \leadsto u\}$ just before the execution of the if at line \ref{line:recursion}. Then $S' \subseteq S \cup \{v\}$ by def. \ref{def:reachability}, because $(v, w) \in E$ and $v.color = \mathrm{GRAY}$ after line \ref{line:v_gray}.
		
		If $v \in nodes(T.c_1)$, then, by def. \ref{def:win_tree}, there would be a $T' \in WinTrees(G, g)$ subtree of $T.c_1$ such that $T'.root = v$ and $nodes(T') \cap S = \emptyset$. But since $w \in G.Adj(v)$, by def. \ref{def:game_graph} $v \ne w$, so $h(T') < h(T.c_1) < h(T)$, contradicting \eqref{eq:min_height_T}, hence $v \notin nodes(T.c_1)$. Therefore, $S' \subseteq S \cup \{v\}$ and $nodes(T.c_1) \cap S = \emptyset$ imply $nodes(T.c_1) \cap S' = \emptyset$.
		
		At line \ref{line:recursion}, we have two cases.
		\begin{itemize}
			\item If $w.color \ne \mathrm{BLACK}$:
			
			by induction, $T.c_1 \in WinTrees(G, g)$ and $nodes(T.c_1) \cap S' = \emptyset$ entail $w.win = \mathrm{TRUE}$ after the execution of $\Call{WinTreeRecursion}{G, g, w}$.
			
			\item If $w.color = \mathrm{BLACK}$:
			
			by the if at line \ref{line:color_if}, $w.win \ne \textbf{NIL}$, so it can only be that $w.win = \mathrm{TRUE}$ or $w.win = \mathrm{FALSE}$. But if $w.win = \mathrm{FALSE}$, then $eval(w) \ne g$ by the corollary \ref{corol:compl_correctness}, so $T.c_1 \notin WinTrees(G, g)$ by the lemma \ref{lemma:correctness_trees}, contradicting $T.c_1 \in WinTrees(G, g)$. Hence, $w.win = \mathrm{TRUE}$.
		\end{itemize}
		
		Anyway, $w.win = \mathrm{TRUE}$ after the execution of the if at line \ref{line:recursion}, therefore $v.win = \mathrm{TRUE}$ by the else at line \ref{line:if_not_gray} and the if at line \ref{line:change_win}.\pagebreak
		
		\item If $turn(v) \ne g$:
		
		by def. \ref{def:win_tree}, the hypothesis implies that for all $w \in G.Adj(v)$ there exists $i_w \in \{1, \dots, T.n\}$ such that $w = T.c_{i_w}.root$ and $T.c_{i_w} \in WinTrees(G, g)$. Since the nodes of $T$ are not in $S$, then $nodes(T.c_{i_w}) \cap S = \emptyset$ and $w.color \ne \mathrm{GRAY}$ at the execution of line \ref{line:if_gray}.
		
		Let $S'_w := \{u \in G.V \mid u.color = \mathrm{GRAY} \text{ and } w \leadsto u\}$ just before the execution of the if at line \ref{line:recursion}. Then $S'_w \subseteq S \cup \{v\}$ by def. \ref{def:reachability}, because $(v, w) \in E$ and $v.color = \mathrm{GRAY}$ after line \ref{line:v_gray}.
		
		Since the height of $T$ is minimal, by the same reasoning of previous case $v \notin nodes(T.c_{i_w})$. Therefore, $S'_w \subseteq S \cup \{v\}$ and $nodes(T.c_{i_w}) \cap S = \emptyset$ imply $nodes(T.c_{i_w}) \cap S'_w = \emptyset$.
		
		At line \ref{line:recursion}, we have two cases.
		\begin{itemize}
			\item If $w.color \ne \mathrm{BLACK}$:
			
			by induction, $T.c_{i_w} \in WinTrees(G, g)$ and $nodes(T.c_{i_w}) \cap S'_w = \emptyset$ entail $w.win = \mathrm{TRUE}$ after the execution of $\Call{WinTreeRecursion}{G, g, w}$.
			
			\item If $w.color = \mathrm{BLACK}$:
			
			by the if at line \ref{line:color_if}, $w.win \ne \textbf{NIL}$, so it can only be that $w.win = \mathrm{TRUE}$ or $w.win = \mathrm{FALSE}$. But if $w.win = \mathrm{FALSE}$, then $eval(w) \ne g$ by the corollary \ref{corol:compl_correctness}, so $T.c_{i_w} \notin WinTrees(G, g)$ by the lemma \ref{lemma:correctness_trees}, contradicting $T.c_{i_w} \in WinTrees(G, g)$. Hence, $w.win = \mathrm{TRUE}$.
		\end{itemize}
		
		Anyway, $w.win = \mathrm{TRUE}$ after the execution of the if at line \ref{line:recursion} for all $w \in G.Adj(v)$, therefore $v.win = \mathrm{TRUE}$ by the else at line \ref{line:if_not_gray} and the if at line \ref{line:change_win}.\qedhere
	\end{itemize}
\end{proof}

\begin{theorem}[Completeness for draw strategies]
	Let $g \in \{1, -1\}$, $v \in G.V$ and $\left<v_0, \dots, v_n \right>$ a path of $G$ such that $v_n = v$, with $n \in \N$. If $v_i.win = \mathbf{NIL}$ for all $i \in \{0, \dots, n\}$ after the execution of $\Call{WinTreeRecursion}{G, g, v_0}$, then $eval(v) \ne g$.
\end{theorem}

\begin{proof}
	By induction on $n$.
	
	Suppose $v_i.win = \mathbf{NIL}$ for all $i \in \{0, \dots, n\}$ after the execution of $\Call{WinTreeRecursion}{G, g, v_0}$. First, note that we can assume without loss of generality that $\left<v_0, \dots, v_n \right>$ is a simple path, because if it were not, we could reduce it to a simple path $\left<v_0, v_{k_1}, \dots, v_{k_m} = v_n \right>$, where $m < n$, by "cutting out" the cycles inside it (i.e. by eliminating vertices inside a cycle and leaving only its endpoint), and apply the theorem's hypothesis to the subpath.
	
	\begin{flushleft}
		BASE: $n = 0$
	\end{flushleft}
	
	Since $v.win = v_0.win = \mathbf{NIL}$, then $eval(v) \ne g$ by the completeness theorem \ref{theo:completeness}.\pagebreak

	\begin{flushleft}
		INDUCTION: $n \ge 1$
	\end{flushleft}
	
	Since $v_0.win = \mathbf{NIL}$, by the if at line \ref{line:change_win} the procedure must have visited each of the vertices adjacent to $v_0$ (not necessarily after visiting it), and since $v_{i+1} \in G.Adj(v_i)$ and $v_{i+1}.win = \mathbf{NIL}$ for all $i \in \{0, \dots, n-1\}$, the same must be true for $v_1$, $v_2$ and so on up to $v_{n-1}$. Hence, since $\left<v_0, \dots, v_n \right>$ is a simple path, by the ifs at line \ref{line:color_if} and \ref{line:recursion} there must be a time when $v_i.color = \mathrm{GRAY}$ for all $i \in \{0, \dots, n-1\}$ during the execution of the procedure.
	
	Suppose by contradiction that $eval(v) = g$. Then there is a $T \in WinTrees(G, g)$ such that $T.root = v$ by the lemma \ref{lemma:completeness_trees}.
	
	Let $S := \{w \in \{v_0, \dots, v_{n-1}\} \mid v \leadsto w\}$. Then $nodes(T) \cap S \ne \emptyset$, otherwise, when $v_i.color = \mathrm{GRAY}$ for all $i \in \{0, \dots, n-1\}$, we have $v_n.win = v.win = \mathrm{TRUE}$ after the execution of $\Call{WinTreeRecursion}{G, g, v}$ by the lemma \ref{lemma:correct_tree_construction}, contradicting the hypothesis.
	
	But, by def. \ref{def:win_tree}, $nodes(T) \cap S \ne \emptyset$ implies that there is a $T' \in WinTrees(G, g)$ subtree of $T$ such that $T'.root \in \{v_0, \dots, v_{n-1}\}$, therefore $eval(v_i) = g$ for some $i \in \{0, \dots, n-1\}$ by the lemma \ref{lemma:correctness_trees}, and this cannot be true by inductive hypothesis.
\end{proof}

This last theorem directly gives us a way to always choose a move that leads at least to a draw (if one exists) against the player $g$, provided that the strategy is adopted starting from the initial position.

Note also that the lemma \ref{lemma:correct_tree_construction} could be used to avoid unnecessary multiple recursive calls, by storing for each vertex the gray vertices reached during the visit of its branch, and then check if some of them are no more gray at the next possible visit of that vertex. However, this may significantly increase memory usage or even the execution time, if repetition of positions occurs very often in the game. Therefore, exploiting this idea may not necessarily lead to optimization.\pagebreak

\nocite{*}
\printbibliography[heading=bibintoc]

@article{shannon:computer_playing_chess,
	author = {Shannon, Claude E.},
	title = {XXII. Programming a Computer for Playing Chess},
	journaltitle = {Philosophical Magazine},
	date = {1950},
	series = {7},
	volume = {41},
	number = {314},
	month = {March},
}

@book{von_neumann:game_theory,
	author = {	
	von Neumann, John and Morgenstern, Oskar},
	title = {Theory of Games and Economic Behavior},
	date = {1944},
	publisher = {Princeton University Press},
}

@book{cormen:algorithms,
	author = {Cormen, T. H. and Leiserson, C. E. and Rivest, R. L. and Stein, C.},
	title = {Introduction to algorithms},
	date = {2009},
	publisher = {Pearson},
	edition = {3},
}

@book{hopcroft:computation,
	author = {Hopcroft, J. E. and Motwani, R. and Ullman, J. D.},
	title = {Introduction to Automata Theory, Languages, and Computation},
	date = {2006},
	publisher = {Pearson},
	edition = {3},
}

@book{odifreddi:recursion_theory,
	author = {Odifreddi, Piergiorgio},
	title = {Classical Recursion Theory},
	subtitle = {The Theory of Functions and Sets of Natural Numbers},
	date = {1989},
	publisher = {Elsevier},
	edition = {1},
	volume = {1},
}

@book{odifreddi:recursion_theory_vol2,
	author = {Odifreddi, Piergiorgio},
	title = {Classical Recursion Theory},
	subtitle = {The Theory of Functions and Sets of Natural Numbers},
	date = {1999},
	publisher = {Elsevier},
	edition = {1},
	volume = {2},
}

@book{arora:computational_complexity,
	author = {Arora, S. and Barak, B.},
	title = {Computational Complexity: A Modern Approach},
	date = {2009},
	publisher = {Cambridge University Press},
	edition = {1},
}

@article{tarjan:dfs_graph,
	author = {Tarjan, Robert},
	title = {Depth-First Search and Linear Graph Algorithms},
	journal = {SIAM Journal on Computing},
	volume = {1},
	number = {2},
	pages = {146-160},
	year = {1972},
}

@article{schaeffer:checkers_solved,
	author = {Schaeffer, Jonathan and Burch, Neil and Björnsson, Yngvi and Kishimoto, Akihiro and Müller, Martin and Lake, Robert and Lu, Paul and Sutphen, Steve},
	title = {Checkers Is Solved},
	journaltitle = {Science},
	date = {2007},
	number = {317},
	month = {september},
	pages = {1518-1522},
	doi = {10.1126/science.1144079},
}

@thesis{allis:connect_four,
	author = {Allis, Victor},
	title = {A Knowledge-based Approach of Connect-Four},
	type = {Masters Thesis},
	institution = {Department of Mathematics and Computer Science Vrije Universiteit},
	date = {1988},
	subtitle = {The Game is Solved: White Wins},
	location = {Amsterdam, The Netherlands},
	month = {october},
}

@online{allen:connect_four,
	author = {Allen, James Dow},
	title = {Expert Play in Connect-Four},
	date = {1990},
	url = {https://tromp.github.io/c4.html},
}

\end{document}